\documentclass{amsart}

\usepackage{graphicx}
\usepackage{latexsym,color}
\usepackage{graphicx}
\usepackage{amssymb}
 \usepackage{amsfonts}
 \usepackage{amsmath}
\usepackage{amsthm}
\usepackage[final]{changes}
\newtheorem{thm}{Theorem}[section]
 
 \newtheorem{lem}[thm]{Lemma}
 \newtheorem{rem}[thm]{Remark}

\newtheorem{prop}[thm]{Proposition}
\newtheorem{asm}[thm]{Assumption}

\numberwithin{equation}{section}

\begin{document}

\title{Market Delay and $G$--expectations}
\thanks{$^*$Corresponding author}

   \author{Yan Dolinsky$^{*\ddag}$ \address{
 School of Mathematical Sciences, Monash University, Melbourne. \\
 e.mail: yan.dolinsky@mail.huji.ac.il}}${}$\\
\author{Jonathan Zouari$^\dagger$ \address{
 Department of Statistics, Hebrew University of Jerusalem, Israel.\\
 e.mail: jonathan.zouari@mail.huji.ac.il}
  ${}$\\
Hebrew University$^{\ddag\dagger}$ and Monash University$^\ddag$}

\date{\today}

\begin{abstract}
  We study super-replication of contingent claims in markets with
  delayed filtration.
   The first result in
this paper reveals that in the Black--Scholes model with constant delay
the super-replication price is prohibitively costly and leads to
trivial buy-and-hold strategies.
Our second result
  says that the scaling limit of super--replication prices for
  binomial models with a fixed number of times of delay $H$ is equal to the $G$--expectation with
  volatility uncertainty interval $[0,\sigma\sqrt{H+1}]$.
  \end{abstract}

\subjclass[2010]{91G10, 91G20, 60F05}
 \keywords{Super--replication, Market delay, Duality, $G$-expectation}%
\maketitle \markboth{Y.Dolinsky and J.Zouari}{Market Delay and $G$--expectations}
\renewcommand{\theequation}{\arabic{section}.\arabic{equation}}
\pagenumbering{arabic}

\section{Introduction}
\label{intro}

This paper deals with super-replication of European options in
financial markets with delayed information.
This corresponds to the case
where there is a time delay in receiving market information (or in applying it), which causes
the trader’s filtration simply being a delayed one in comparison to a price filtration.

Although the topic of
hedging with delay and restricted information
enjoyed a considerable
attention in the literature
(see, for instance, \cite{CKT,DPR,F,IM, Ka,KS,KX,MTT,S,SZ}), to the best of
our knowledge, super--replication in this setup was studied only in \cite{CKT,IM}.
Recently, in \cite{CKT} the authors studied
the fundamental theorem of asset pricing and the super--replication problem in a
general continuous time two filtration setting. Their setting includes the delay setup.

We adopt the setup from \cite{CKT} (for delay) and study in details the super--replication price in
the Black--Scholes model and the continuous time limit
of the binomial models.
Our first result says that for the Black--Scholes model with a constant delay $h>0$,
the cheapest way to super-replicate a vanilla option is to apply
a trivial buy-and-hold strategy.
The main idea of the proof is to use the duality theory which was developed in
\cite{CKT} and the Girsanov theorem in order to
construct a rich enough family of pricing measures.

Next, we overcome this negative result by considering scaling limits
of super-replication prices in the binomial models.
We fix a natural number $H$ and assume that
the delay in the received information equals to $H$ trading times.
Thus, we let the continuous time delay $\frac{H}{n}$ tend to zero linearly in the time step.
With this type of scaling we prove that the
super--replication prices in the binomial models with constant volatility $\sigma>0$ converge
to the $G$--expectation of Peng \cite{P1} with uncertainty interval $[0,\sigma \sqrt{H+1}]$.
We prove this result under quite general assumptions which allow to consider path dependent payoffs.
In particular, when the payoff is convex (can be path dependent) we converge to the Black--Scholes price
with increased volatility $\sigma \sqrt{H+1}$.

The above result is inspired by a recent work of Ichiba and Mousavi \cite{IM}
where the authors considered super--replication in binomial
models with delay. The authors setup was a bit different
and was restricted to contingent claims with convex and path--independent payoffs.
Our approach is to apply the duality theory
from \cite{CKT} and analyze
the asymptotic behaviour
of the corresponding pricing measures (the dual objects).
This approach allows to treat a more general setup (than in \cite{IM}) and provides an additional intuition for the
limit.

It is important to mention the paper \cite{AHMP} where the authors model the risky asset via
a non-linear stochastic differential delay equation in a Brownian framework.
The main difference is that in their model the filtration is generated by the risky asset, and so the corresponding
financial market is complete.
Of course this is not true in our case.

The paper is organized as follows. In Section 2 we
show that in a Black--Scholes model with a fixed delay,
trivial buy-and-hold strategies yield optimal super--replication.
In Section 3 we give the scaling limit of super--replication
prices
in the binomial models with vanishing delay.

\section{The Black--Scholes model with constant delay}\label{sec:7}\setcounter{equation}{0}
\subsection{Preliminaries and the buy--and--hold propery}
Consider a complete probability space
$(\Omega^W,\mathcal F^W,\mathbb P^W)$
together with a standard one dimensional Brownian motion
$W=\{W_t\}_{t=0}^\infty$
and the filtration $\mathcal F^W_t = \sigma\{W_u| u\leq t\}$ completed
by the null sets. Our Black--Scholes financial market consists of a safe asset $B$ used as
numeraire, hence $B\equiv 1$, and of a risky asset $S$ whose value at time $t$ is given by
\begin{equation*}
S_t=s e^{\sigma W_t+\mu t}, \ \ s>0
\end{equation*}
where $\sigma>0$ is called volatility and $\mu\in\mathbb{R}$ is
another constant called the drift.

Next, we describe the super--replication setup from \cite{CKT} for the delay setup.
Let $T=1$ be the horizon of our financial market.
We fix a constant delay parameter $h>0$ and consider a setup where the control of the investor at time $t$
can depend only on the information observed before time $t-h$. In this setup
a simple trading strategy is a stochastic process
$$\gamma_t=\sum_{i=1}^l \gamma_i\mathbb{I}_{(t_{i-1},t_i]}$$
where $0=t_0<t_1<...<t_l=1$ is a determinisitc partition and
for any $i$, $\gamma_i$ is a random varible $\mathcal F^W_{(t_{i-1}-h)^{+}}$ measurable. The corresponding portfolio value at the maturity date equals to
$$Y^{\gamma}_1=\sum_{i=1}^{l} \gamma_{t_{i-1}}(S_{t_i}-S_{t_{i-1}}).$$
We denote by $\mathcal C$ the convex cone of all super--replicable claims
(by simple strategies), that is
$$\mathcal C=\{Y^{\gamma}_1: \ \gamma \ \mbox{is} \ \mbox{simple}\}-L^0_{+}(\mathcal F^W_1,\mathbb P^W).$$
For a given $p>1$ and a European contingent claim $X\in L^p(\mathcal F^W_1,\mathbb P^W)$ we define the super--hedging price by
$$V(X)=\inf\{x: \exists U\in \overline{\mathcal C_p} \ \mbox{such} \ \mbox{that} \ x+U\geq X \ \mbox{a.s.}\}$$ where
$\overline{\mathcal C_p}$ is the closure in $L^p(\mathcal F^W_1,\mathbb P^W)$
of the set $\mathcal C_p:=\mathcal C\cap L^p(\mathcal F^W_1,\mathbb P^W)$.

From the duality result
 Theorem 4.4 in \cite{CKT} we obtain that
 \begin{equation}\label{tei}
 V(X)=\sup_{\mathbb Q\in \mathbb M^q} \mathbb E_{\mathbb Q}[X]
 \end{equation}
 where $\mathbb M^q$ is the set of all equivalent probability measures
 $\mathbb Q\sim\mathbb P^W$ such that
 $\frac{d\mathbb Q}{d\mathbb P^W}{|\mathcal F^W_1}\in L^q(\mathcal F^W_1,\mathbb P^W)$
 (where $\frac{1}{q}+\frac{1}{p}=1$) and
 $$\mathbb E_{\mathbb Q}\left(S_{T_2}-S_{T_1}|\mathcal F^W_{(T_1-h)^{+}}\right)=0
 \ \mbox{for} \ \mbox{all} \ T_2\geq T_1.$$

 We arrive to the first result of the paper.
 \begin{thm}\label{thm7.0}
 Let $X=f(S_1)$ where
 $f:\mathbb R_{+}\rightarrow\mathbb R_{+}$ is a nonnegative
 and lower semi--continuous function. Then the super-replication price of $X$
satisfies
$$V(X)=\hat f(s)$$
where $\hat f$ denotes
the concave envelope of $f$ (the joint value can be equal to $\infty$). Moreover,
if $\hat f<\infty$ then an optimal (buy--and--hold) strategy exists
and is explicitly defined by
$$\gamma\equiv \partial_{+}\hat{f}(s).$$
\end{thm}
\begin{rem}
By the cash-invariance property of $V(X)$,
the condition that
$f:\mathbb R_{+}\rightarrow\mathbb R_{+}$ is nonnegative could be relaxed by the requirement to be bounded from below.
\end{rem}
 \subsection{Proof of Theorem \ref{thm7.0}}
 In this section we prove Theorem \ref{thm7.0}.
 We start with the following auxiliary result.
 \begin{lem}\label{lem7.1}
Let $\nu>0$ be a given constant. There exists a sequence of probability measures
$\tilde{\mathbb Q}_n\sim\mathbb P^W$, $n\in\mathbb N$ such that:
\begin{equation}\label{7.1}
\mathbb E_{\tilde{\mathbb Q}_n}[S_t]=s, \ \forall  t\in [0,1]
\end{equation}
and
\begin{equation}\label{7.2}
(S,\tilde{\mathbb Q}_n)\Rightarrow (\{s e^{\nu W_t-\nu^2 t/2}\}_{t=0}^1,\mathbb P^W).
\end{equation}
The relation given by (\ref{7.2}) means that
the distribution of
$\{S_t\}_{t=0}^1$ under
$\tilde{\mathbb Q}_n$ converges to
the distribution of $\{s e^{\nu W_t-\nu^2 t/2}\}_{t=0}^1$ under $\mathbb P^W$.
The weak convergence is on the space of continuous functions $C([0,1];\mathbb R)$ equipped
with the topology of uniform convergence.
\end{lem}
\begin{proof}
The proof will be done in two steps.\\
\textbf{Step I:} In this step we prove (\ref{7.1}).
 Fix $n\in\mathbb N$. Set,
$\Psi(z):=-1\vee(z\wedge 1)$, $z\in\mathbb R$. Introduce the function $g^{(n)}:[0,1]\rightarrow\mathbb R$
\begin{eqnarray*}
&g^{(n)}_t:=\mathbb E_{\mathbb P^W}\left(\exp\left(\sigma W_{t}+(\nu-\sigma)
\sum_{i=1}^{[nt]-1} \Psi\left(W_{\frac{i}{n}}- W_{\frac{i-1}{n}}\right)+\right.\right.\nonumber\\
&\left.\left.(nt-[nt])(\nu-\sigma)\Psi\left(W_{\frac{[nt]}{n}}- W_{\frac{[nt]-1}{n}}\right)-\frac{\nu^2 t}{2}\right)\right)
\end{eqnarray*}
where $[\cdot]$ is the integer part of $\cdot$. Observe that $g^{(n)}$ is continuous.

Next, define the processes
$\{\tilde W^{(n)}_t\}_{t=0}^1$, $\{\kappa^{(n)}_t\}_{t=0}^1$ by the following recursive relations
\begin{eqnarray*}
&\tilde W^{(n)}_t:=W_t+\frac{\ln g^{(n)}_t}{\sigma}-\int_{0}^t \kappa^{(n)}_u du, \\
&\kappa^{(n)}_t:=\frac{1}{\sigma}
\left(n(\nu-\sigma)\sum_{i=0}^{n-1}
\Psi\left(\tilde W^{(n)}_{\frac{i}{n}}-\tilde W^{(n)}_{\frac{i-1}{n}}\right)
-\mu-\frac{\nu^2}{2}\right)\mathbb{I}_{\left(\frac{i}{n},\frac{i+1}{n}\right)}.
\end{eqnarray*}
We notice (by taking conditional expectation given $\mathcal F^W_{[nt]/n}$) that
\begin{eqnarray*}
&g^{(n)}_t=\exp\left(\sigma^2 (t-[nt]/n)/2-\nu^2 t/2\right)\times\label{7.3}\\
&\mathbb E_{\mathbb P^W}\left(\exp\left(\sigma W_{\frac{[nt]}{n}}+(\nu-\sigma)\nonumber
\sum_{i=1}^{[nt]-1} \Psi\left(W_{\frac{i}{n}}- W_{\frac{i-1}{n}}\right)+\right.\right.\nonumber\\
&\left.\left.(nt-[nt])(\nu-\sigma)\Psi\left(W_{\frac{[nt]}{n}}- W_{\frac{[nt]-1}{n}}\right)\right)\right)\nonumber
\end{eqnarray*}
and so (recall that $|\Psi|\leq 1$) $\ln g^{(n)}$
is differentiable for all $t\in [0,1]\setminus\{0,\frac{1}{n},\frac{2}{n},...,1\}$ and the derivative is bounded.
Furthermore,
$\{\kappa^{(n)}_t\}_{t=0}^1$ is a uniformly bounded process (again, $|\Psi|\leq 1$). Thus, from
the Girsanov theorem there exists a probability measure
$\tilde{\mathbb Q}_n\sim\mathbb P^W$ such that
$\{\tilde W^{(n)}_t\}_{t=0}^1$ is a Brownian motion under $\tilde{\mathbb Q}_n$.

Hence,
\begin{eqnarray*}
&\mathbb E_{\tilde {\mathbb Q}_n}[S_t]=s\mathbb E_{\tilde{\mathbb Q}_n}\left(\exp\left(\sigma \tilde W^{(n)}_{t}+(\nu-\sigma)
\sum_{i=1}^{[nt]-1} \Psi\left(\tilde W^{(n)}_{\frac{i}{n}}- \tilde W^{(n)}_{\frac{i-1}{n}}\right)+\right.\right.\\
&\left.\left.(nt-[nt])(\nu-\sigma)\Psi\left(\tilde W_{\frac{[nt]}{n}}- \tilde W_{\frac{[nt]-1}{n}}\right)-\ln g^{(n)}_t-\frac{\nu^2 t}{2}\right)\right)=s
\end{eqnarray*}
as required.\\
\textbf{Step II:}
In this step we prove (\ref{7.2}).
Introduce the stochastic process
$$X^{(n)}_t=\ln S_t+\ln  g^{(n)}_t-\ln s, \ t\in [0,1].$$
In order to prove (\ref{7.2})
it is sufficient to show that
\begin{equation}\label{7.2+}
\lim_{n\rightarrow\infty}\mathbb E_{\tilde{\mathbb Q}_n}\left(\sup_{0\leq t\leq 1}|X^{(n)}_t-(\nu \tilde W^{(n)}_t-\nu^2 t/2)|\right)=0
\end{equation}
and
\begin{equation}\label{7.2++}
\lim_{n\rightarrow\infty}|g^{(n)}_t-1|=0.
\end{equation}
Clearly, $|z-\Psi(z)|\leq |z|\mathbb{I}_{|z|>1}$.
Hence,
\begin{eqnarray}\label{7.new}
&\sup_{0\leq t\leq 1}|X^{(n)}_t-(\nu \tilde W^{(n)}_t-\nu^2 t/2)|\leq\\
&(\nu+\sigma) \sum_{k=1}^n |\tilde W^{(n)}_{\frac{k}{n}}- \tilde W^{(n)}_{\frac{k-1}{n}}|
\mathbb{I}_{|\tilde W^{(n)}_{\frac{k}{n}}- \tilde W^{(n)}_{\frac{k-1}{n}}|>1} +\nonumber\\
&2(\nu+\sigma)\sup_{0\leq t\leq 1} |\tilde W^{(n)}_{t}-\tilde W^{(n)}_{\frac{[nt]}{n}}|.\nonumber
\end{eqnarray}
Observe that under the probability measure $\tilde{\mathbb Q}_n$, $\tilde W^{(n)}_{\frac{k}{n}}- \tilde W^{(n)}_{\frac{k-1}{n}}\sim \frac{1}{\sqrt n} N(0,1)$.
Thus from (\ref{7.new}) we obtain
\begin{eqnarray*}
&\lim\sup_{n\rightarrow\infty}\mathbb E_{\tilde{\mathbb Q}_n}\left(\sup_{0\leq t\leq 1}|X^{(n)}_t-(\nu \tilde W^{(n)}_t-\nu^2 t/2)|\right)\leq\\
&(\nu+\sigma)\lim\sup_{n\rightarrow\infty} \left( \sqrt{n} \int_{\sqrt n}^{\infty} 2 z\frac{e^{-z^2/2}}{\sqrt{2\pi} }dz\right)=0
\end{eqnarray*}
 and (\ref{7.2+}) follows.

Finally, we prove (\ref{7.2++}).
From (\ref{7.1}) we have $g^{(n)}_t= \mathbb E_{\tilde{\mathbb Q}_n}[e^{X^{(n)}_t}]$. This together with the trivial equality
$\mathbb E_{\tilde {\mathbb Q}_n}[e^{\nu\tilde W^{(n)}_t-\nu^2 t/2}]=1$ yields
 \begin{eqnarray*}
 &\lim\sup_{n\rightarrow\infty}\sup_{0\leq t\leq 1}\left|g^{(n)}_t-1\right|\leq\\
 &\lim\sup_{n\rightarrow\infty}\mathbb E_{\tilde {\mathbb Q}_n}\left(\sup_{0\leq t\leq1} \left|e^{X^{(n)}_t}-e^{\nu \tilde W^{(n)}_t-\nu^2 t/2}\right|\right)=0
 \end{eqnarray*}
where the last equality follows from (\ref{7.2+}) and uniform integrability of
$\sup_{0\leq t\leq 1}e^{X^{(n)}_t}$, $n\in\mathbb N$.
We conclude that (\ref{7.2++}) holds true as required.
\end{proof}
Next, let $\Gamma$ be the set of all uniformly bounded processes $\alpha=\{\alpha_t\}_{t=0}^1$ which are progressively
measurable with respect to $\mathcal F^W$.
For any $\alpha\in\Gamma$ we introduce the corresponding stochastic exponential
$$S^{(\alpha)}_t=s\exp\left(\int_{0}^t\alpha_u dW_u-\frac{1}{2}\int_{0}^t\alpha^2_u du\right), \ \ t\in [0,1].$$
The following Proposition is a direct application of Lemma \ref{lem7.1} and will be the corner stone in the proof of Theorem \ref{thm7.0}.
\begin{prop}\label{prop7.1}
For any $\alpha\in\Gamma$ there exists a sequence of probability measures
$\mathbb Q_n\in\mathbb M^q$, $n\in\mathbb N$ (the set $\mathbb M^q$ is defined before Theorem \ref{thm7.0})
such that
$$(S,\mathbb Q_n)\Rightarrow (S^{(\alpha)},\mathbb P^W).$$
\end{prop}
\begin{proof}
Let $\Gamma_c\subset\Gamma$ be the set of all processes $\{\alpha_t\}_{t=0}^1$
of the form
\begin{equation}\label{4.102}
\alpha_t = \sum_{j=0}^{J-1} \rho_j( S^{(\alpha)}_{t_1},\dots, \ S^{(\alpha)}_{t_j}) 1_{(t_j,t_{j+1}]}
  \end{equation}
for some times $0=t_0<t_1<\dots<t_J=1$, $\epsilon>0$, and continuous bounded
functions $\rho_j:\mathbb R^{j} \rightarrow [\epsilon,\infty)$. Observe that $\rho_0$ is a constant.
Without loss of generality we assume that
the mesh of the partition is less than the constant delay $h>0$, i.e.
$t_i<t_{i-1}+h$ for all $i=1,...,J$.

By applying Levy's theorem in
a similar way to Lemma 4.6 in \cite{BD} we arrive that $M= S^{(\alpha)}$ is the unique (in law) martingale with initial value $M_0=s$ which satisfies
\begin{equation}\label{4.103}
\left\{\sum_{j=0}^{J-1} \int_{t_j\wedge t }^{t_{j+1}\wedge t }\frac{dM_u}{\rho_j(M_{t_1},...,M_{t_j})M_u}\right\}_{t=0}^1
 \ \ \mbox{is} \ \mbox{a} \ \mbox{standard} \ \mbox{Brownian} \ \mbox{motion}.
\end{equation}
Standard density arguments (see Lemma 3.4 in \cite{BDP} for the case $d=1$) imply that the set
of distributions $\left(\{S^{(\alpha)}\}_{t=0}^1,\mathbb P^W\right)_{\alpha\in\Gamma_c}\subset
\left(\{S^{(\alpha)}\}_{t=0}^1,\mathbb P^W\right)_{\alpha\in\Gamma}$ is dense.
Thus, it sufficient to prove the Proposition for $\alpha\in\Gamma_c$.

Fix $i=0,...,J-1$. Consider the Brownian motion
$W^{(i)}_t:=W_{t_i+t}-W_{t_i}$, $t\geq 0$.
Since $W^{(i)}$ is independent of $\mathcal F^W_{t_i}$, we can apply Lemma \ref{lem7.1}
for the Brownian motion $W^{(i)}$ and the volatility $\nu:=\rho_i( S^{(\alpha)}_{t_1},\dots, S^{(\alpha)}_{t_i})$.
Only we restrict the corresponding
measures $\tilde{\mathbb Q}_n$ to the interval $[0,t_{i+1}-t_i]$.
By following this procedure for all $i=0,1...,J-1$ (we shift the interval $[0,t_{i+1}-t_i]$ to the interval $[t_i,t_{i+1}]$)
we obtain a sequence of probability measures
${\mathbb Q}_n$ which satisfy the following properties (analogous to (\ref{7.1})--(\ref{7.2}))
\begin{equation}\label{7.8}
\mathbb E_{{\mathbb Q}_n}[S_t|\mathcal F^W_{t_i}]=S_{t_i}, \ \forall  t\in [t_i,t_{i+1}], \ \ i=0,...,J-1
\end{equation}
and
\begin{equation}\label{7.9}
\left(\left\{\frac{S_t}{S_{t_i}}\right\}_{t=t_i}^{t_{i+1}},{\mathbb  Q}_n\right)\Rightarrow \left(\left\{\frac{S^{(\alpha)}_t}{S^{(\alpha)}_{t_i}}\right\}_{t=t_i}^{t_{i+1}},\mathbb P^W\right), \ \ i=0,1,...,J-1.
\end{equation}
Clearly, (\ref{7.9}) implies $(S,\mathbb Q_n)\Rightarrow (S^{(\alpha)},\mathbb P^W)$.

It remains to argue that
$\mathbb Q_n\in\mathbb M^q$ for all $n\in\mathbb N$ .
Fix $n$.
Observe that there exists a drift process
$\{\lambda^{(n)}_t\}_{t=0}^1$ such that
$\{W_t-\int_{0}^t\lambda^{(n)}_u du\}_{t=0}^1$ is a Brownian motion under $\mathbb Q_n$.
Moreover, for any $i=0,...,J-1$ the drift $\lambda^{(n)}_{|(t_i,t_{i+1}]}$ can be computed
by applying Lemma \ref{lem7.1}
for the Brownian motion $W^{(i)}_t:=W_{t_i+t}-W_{t_i}$, $t\geq 0$
and the volatility $\nu:=\rho_i( S^{(\alpha)}_{t_1},\dots, S^{(\alpha)}_{t_i})$.
Since the functions $\rho_i$, $i=0,1,...,J$ are uniformly bounded we conclude that the drift process
$\{\lambda^{(n)}_t\}_{t=0}^1$ is uniformly bounded and so
$\frac{d\mathbb Q_n}{d\mathbb P^W}{|\mathcal F^W_1}
\in L^q(\mathcal F^W_1,\mathbb P^W)$, recall that $q=\frac{p}{p-1}\in [1,\infty)$.

Finally, let $0\leq T_1<T_2\leq 1$. Set,
$k:=\max\{i: t_i\leq T_1\}$. From (\ref{7.8}) it follows that
$$\mathbb E_{{\mathbb Q}_n}[S_{T_2}|\mathcal F^W_{t_k}]=\mathbb E_{{\mathbb Q}_n}[S_{T_1}|\mathcal F^W_{t_k}]=S_{t_k}.$$
Since $t_k> T_1-h$ (recall that $t_i-t_{i-1}<h$ for all $i$)
we conclude that
$$\mathbb E_{\mathbb Q}\left(S_{T_2}-S_{T_1}|\mathcal F^W_{(T_1-h)^{+}}\right)=0$$
and the proof is completed.
\end{proof}
Now, we are ready to prove Theorem \ref{thm7.0}.
\begin{proof}
Proof of Theorem \ref{thm7.0}.\\
First, if $\hat f<\infty$ then from concavity we obtain
$$\hat f(s)+\partial_{+}\hat{f}(s)(S_1-s)\geq\hat f(S_1)\geq f(S_1) \ \mbox{a.s.}$$
In other words the trivial strategy with initial capital $\hat f(s)$ and
$\gamma\equiv \partial_{+}\hat{f}(s)$
is a super--hedge.

Thus, it in order to complete the proof it remains to establish that
$V(X)\geq \hat f(s)$.
To that end,
recall the set $\Gamma$ which is defined before Proposition \ref{prop7.1}.
From the Fatou lemma, the duality given by (\ref{tei}) and Proposition \ref{prop7.1} we obtain
\begin{equation}
V(X)\geq\sup_{\alpha\in\Gamma} \mathbb E_{\mathbb P^W}[f(S^{(\alpha)}_1].
\end{equation}
Moreover,
from Lemmas 3.2--3.3 in \cite{N}.
we have
$$\sup_{\alpha\in\Gamma} \mathbb E_{\mathbb P^W}[f(S^{(\alpha)}_1]\geq\hat f(s).$$
Thus we conclude that
$V(X)\geq\hat f(s)$.
\end{proof}
\begin{rem}
Roughly speaking, Proposition \ref{prop7.1}
says that any local volatility model is a cluster point of the probability measures
in $\mathbb M^q$ (the dual objects).
In order that the super--replication price will be "reasonable"
we need to have a uniform bound on the local volatility.
In the next section we consider a continuous time limit of binomial models.
We prove that if we scale
the delay in the "right" way, then the local volatility models which appear in the
limit are uniformly bounded.
\end{rem}
\section{Scaling limit of binomial models with vanishing delay}
\subsection{The main result}
Let $\bar\Omega={\{-1,1\}}^{\mathbb{N}}$ be the space of infinite sequences
$\omega=(\omega_1,\omega_2,...)$;
$\omega_i\in\{-1,1\}$ with the product probability
$\mathbb{P}=\{\frac{1}{2},\frac{1}{2}\}^{\mathbb{N}}$.
Define the canonical sequence of independent and identically distributed (i.i.d.) random variables $\xi_1,\xi_2,...$ by
$\xi_i(\omega)=\omega_i$, ${i}\in\mathbb{N}$,
and consider the natural filtration
$\mathcal{F}_k=\sigma{\{\xi_1,...,\xi_k}\}$, $k\geq 1$
and let $\mathcal{F}_0$ be trivial.

Next, we introduce a sequence of binomial models
with volatility $\sigma>0$. For any $n$ consider the $n$--step
binomial model of a financial market which is active at times
$0,1/n,2/n,\ldots,1$. We assume that the market consists of
a safe asset$\equiv 1$ used as
numeraire and of a stock. The stock price at time
$k/n$ is given by
\begin{equation*}
S^{(n)}_k=se^{\sigma\sqrt\frac{1}{n}\sum_{i=1}^{k} \xi_i},
 \ \ \ k=0,1,...,n,
\end{equation*}
Similarly to \cite{IM} we fix a natural number $H\in\mathbb N$ and consider a situation where there is a delay of $H$
trading times.
Thus, we consider a sequence of binomial models with vanishing delay
$\frac{H}{n}$, $n\in\mathbb N$.

In the $n$--step binomial model a self--financing
portfolio $\pi$ with an initial capital $x$
is a pair $\pi=(x,\{\gamma_k\}_{k=0}^{n-1})$ where
for any $k$,
$\gamma_k$ is a $\mathcal{F}_{(k-H)^{+}}$--measurable
random variable.
The corresponding portfolio value at the maturity date equals to
$$Y^\pi_1=x+\sum_{i=0}^{n-1} \gamma_i(S^{(n)}_{i+1}-S^{(n)}_{i}).$$
We are interested in super--replication of European contingent claims.
Formally, given
a random variable $F_n$ (which represents the payoff of the claim) which is $\mathcal F_n$ measurable,
the super--replication price is given by
\begin{equation*}
V_n:=V_n(F_n)=\inf\left\{x:  \exists \pi=(x,\gamma) \
\mbox{such}  \ \mbox{that}\  Y^\pi_1\geq F_n,\
 \ \mathbb{P}\mbox{-a.s.}
\right\}.
\end{equation*}
Again, by applying Theorem 4.4 in \cite{CKT} (for the relatively simple case of finite probability space and discrete trading) we get
\begin{equation}\label{dual}
V_n=\sup_{\mathbb Q\in\mathcal Q^H_n}\mathbb E_{\mathbb Q}[F_n]
\end{equation}
where $\mathcal Q^H_n$ is
the set of all probability measures $\mathbb Q\sim\mathbb P$ on $(\bar\Omega,\mathcal F_n)$ for which
\begin{equation}\label{3.1}
\mathbb E_{\mathbb Q}(S^{(n)}_{k+1}-S^{(n)}_k|\mathcal F_{(k-H)^{+}})=0, \ \ \forall k=0,1...,n-1.
\end{equation}

The main result of this paper is the identification of the limit (as time step goes to zero)
of the super--replication prices as a $G$--expectation in the sense of Peng \cite{P1}.
Formally, let $\Omega=C([0,1], \mathbb R)$
be the space of
continuous paths equipped with the topology of uniform convergence
and the Borel $\sigma$--field $\mathbb F=\mathcal B(\Omega$). We denote by $\mathbb B=\mathbb B_t$, $t\geq 0$
the canonical process $\mathbb B_t(\omega)=\omega_t$.
Introduce the following set of probability measures on $(\Omega,\mathbb F)$
$$\mathcal Q^H:=\{\mathbb Q: \ \mathbb B \ \mbox{is} \ \mbox{a} \ \mathbb Q\mbox{-martingale}
\ \mbox{with} \ \mathbb B_0=0,\ d\langle\mathbb B\rangle/dt\leq \sigma^2(H+1) \ \mathbb Q\otimes dt\mbox{-a.s.}\}
$$
We assume the following.
\begin{asm}\label{asm4.1}
Let
$F:\Omega\rightarrow\mathbb{R}_{+}$
be a continuous map such that there exists
a constants $C,p>0$ for which
\begin{equation*}
\begin{split}
F(\omega)\leq C(1+\|\omega\|_\infty^p), \
\ \ \forall{\omega}\in \Omega.
\end{split}
\end{equation*}
For any $n\in\mathbb{N}$, let
$\mathcal{W}_n:\mathbb{R}^{n+1}\rightarrow\Omega$
be the linear interpolation operator given by
$$
\mathcal{W}_{n}(y)(t):=
\left(\left[{nt}\right]+1-{nt}\right)y_{\left[{nt}\right]}+
\left(nt-\left[nt\right]\right)y_{\left[nt\right]+1},
 \ \ \forall{t}\in[0,1]
$$
where $y=(y_0,y_1,...,y_n)\in\mathbb{R}^{n+1}$ and $[\cdot]$ denotes the integer part of $\cdot$.
In the $n$--step binomial model the payoff of the European contingent
claim is given by
\begin{equation*}
F_n:=F\left(\mathcal{W}_n(S^{(n)})\right)
\end{equation*}
where, by definition we consider, $\mathcal{W}_n(S^{(n)})$ as a random
element with values in $\Omega$.
\end{asm}
Next, we formulate our main result.
\begin{thm}\label{thm4.1}
Assume Assumption \ref{asm4.1}. Then the limit of the super--replication prices is given by
\begin{equation*}
\lim_{n\rightarrow\infty} V_n=\sup_{\mathbb Q\in\mathcal Q^H}\mathbb E_{\mathbb Q}[F(\mathbb S)]
\end{equation*}
where $\mathbb S$ is the exponential martingale
$$\mathbb S_t=s\exp(\mathbb B_t-\langle\mathbb B_t\rangle/2), \ \ t\in [0,1].$$
\end{thm}
\subsection{Proof of the Upper Bound}
In this section we prove that
\begin{equation}\label{4.2}
\lim \sup_{n\rightarrow\infty} V_n\leq \sup_{\mathbb Q\in\mathcal Q^H}\mathbb E_{\mathbb Q}[F(\mathbb S)].
\end{equation}
Without loss of generality (by passing to a subsequence) we assume that
$\lim_ {n\rightarrow\infty} V_n$ exists (it may be $\infty$).
From the duality given by (\ref{dual}) it follows that for any $n\in\mathbb N$ there exists a probability measure
$\mathbb Q_n\in\mathcal Q^H_n$ such that
\begin{equation}\label{4.3}
V_n<\frac{1}{n}+ \mathbb E_{\mathbb Q_n}[F_n].
\end{equation}
For any $n\in\mathbb N$ introduce the martingale $\{M^{(n)}_k\}_{k=0}^n$
\begin{equation*}
M^{(n)}_k:=\mathbb E_{\mathbb Q_n}(S^{(n)}_{k}|\mathcal F_{(k-H)^{+}}), \ \ k=0,1,...,n
\end{equation*}
The fact that $M^{(n)}$ is a martingale follows from (\ref{3.1}).
Clearly, there exists a constant $C$ such that
\begin{equation}\label{4.4+}
|M^{(n)}_k- S^{(n)}_k|\leq \frac{C}{\sqrt n}  S^{(n)}_k \ \ \forall k\leq n.
\end{equation}
By applying Lemma 3.3 in \cite{BDP}
we obtain that $(\mathcal{W}_{n}(M^{(n)})|\mathbb Q_n)$ is tight on the space $\Omega=C([0,1], \mathbb R)$ and
\begin{equation}\label{4.5+}
\sup_{n\in\mathbb N}\mathbb E_{\mathbb Q_n}[\max_{0\leq k\leq n}M^{(n)}_k+|\ln M^{(n)}_k|]^{2m}<\infty, \ \ \forall m>0.\\
\end{equation}
 So there exists a subsequence (which, for ease of notation, we still
index by $n$) which converges in distribution to a continuous stochastic process
$M=\{M_t\}_{t=0}^1$.
The fact that $\{M^{(n)}_k\}_{k=0}^n$ is a martingale for all $n$ and
(\ref{4.5+}) implies (see the standard arguments after Lemma 3.3 in \cite{BDP}) that
$M$ is a strictly positive martingale.

Next, from Assumption \ref{asm4.1} and (\ref{4.3})--(\ref{4.5+})
we obtain
$$\lim_{n\rightarrow\infty}V_n\leq \mathbb E[F(M)].$$
Introduce the continuous local martingale
$N_t:=\int_{0}^t\frac{d M_u}{M_u}$, $t\in [0,1]$.
Clearly,
$$M_t=\exp(N_t-\langle N_t\rangle/2), \ \ t\in [0,1].$$

We conclude that in order to establish (\ref{4.2}), it remains to prove the following lemma.
\begin{lem}
The quadratic variation of the local martingale $N$ satisfies
\begin{equation}\label{4.6}
\frac{d\langle N\rangle}{dt}\leq \sigma^2(H+1) \ \ \forall t\in [0,1], \ \ \mbox{a.s.},
\end{equation}
which in particular implies that $N$ is a (true) martingale.
\end{lem}
\begin{proof}
The implication that $N$ is a martingale is clear form the Burkholder--Davis--Gundy inequality. Thus, let us prove
(\ref{4.6}).

Fix $n\in\mathbb N$.
From the Taylor expansion for $e^x$ we get that for any $k\geq H$,
$$ S^{(n)}_k- S^{(n)}_{k-1}= S^{(n)}_{k-H-1}\left(1+\frac{\sigma\xi_{k}}{\sqrt n}+O(1/n)\right)$$
where the term $O(1/n)$ is uniformly bounded by $c/n$ for some constant $c$. This together with (\ref{3.1}) gives
(recall that $\mathbb Q_n\in\mathcal Q^H_n$)
\begin{equation}\label{4.7}
\mathbb E_{\mathbb Q_n}(\xi_{k}|\mathcal F_{k-H-1})=O(1/\sqrt n), \ \ k=H,...,n.
\end{equation}
For any $j=0,...,H$ introduce the
stochastic processes
$A^{n,j}=\{A^{n,j}_k\}_{k=0}^n$ and $M^{n,j}=\{M^{n,j}_k\}_{k=0}^n$
by
\begin{eqnarray*}\label{4.8}
&A^{n,j}_k=\frac{\sigma}{\sqrt n}\sum_{i=1}^{[k/{(H+1)}]}\mathbb E_{\mathbb Q_n}\left(\xi_{i (H+1)+j}|\mathcal F_{(i-1) (H+1)+j}\right),\\
&M^{n,j}_k=\frac{\sigma}{\sqrt n}\left(\sum_{i=1}^{[k/{(H+1)}]}\xi_{i (H+1)+j}-
\mathbb E_{\mathbb Q_n}\left(\xi_{i (H+1)+j}|\mathcal F_{(i-1) (H+1)+j}\right)\right)
\end{eqnarray*}
where
as before $[\cdot]$ is the integer part of $\cdot$.

Fix $j$. Clearly, $M^{n,j}$ is a martingale with respect to the filtration generated by
$A^{n,j}$ and $M^{n,j}$. Notice that for any $i$,
$M^{n,j}_{i (H+1)}=M^{n,j}_{i (H+1)+1}=...=M^{n,j}_{i (H+1)+H}$. Moreover,
(\ref{4.7}) implies that
\begin{equation}\label{variation}
\mathbb E_{\mathbb Q_n}\left((M^{n,j}_{(i+1) (H+1)}-M^{n,j}_{i (H+1)})^2|\mathcal F_{i(H+1)+j}\right)=
\frac{\sigma^2}{n}+O(n^{-3/2}).
\end{equation}
Thus, from the Martingale Central Limit Theorem (Theorem 7.4.1 in \cite{EK})
it follows that
\begin{equation}\label{4.9}
\left(\{M^{n,j}_{[nt]}\}_{t=0}^1,\mathbb Q_n\right)\Rightarrow\frac{\sigma}{\sqrt{H+1}} W,
\end{equation}
where $W=\{W_t\}_{t=0}^1$ is a standard Brownian motion.
From (\ref{4.7}) we obtain that the sequence
$\left(\{A^{n,j}_{[nt]}\}_{t=0}^1,\mathbb Q_n\right)$, $n\in\mathbb N$ is tight (on the Skorokhod space of right continuous with left limit functions),
and any cluster point is a Lipschitz continuous process.
We conclude that there exists a subsequence (which, for ease of notation, we still
index by $n$) such that we have the convergence of the joint distributions
\begin{equation}\label{4.9+}
\left(\left(\{M^{n,j}_{[nt]}\}_{t=0}^1,\{A^{n,j}_{[nt]}\}_{t=0}^1\right),\mathbb Q_n\right)\Rightarrow (W^{(j)},A^{(j)})
\end{equation}
where $A^{(j)}=\{A^{(j)}_t\}_{t=0}^1$ is a Lipschitz continuous process and
$W^{(j)}=\{W^{(j)}_t\}_{t=0}^1$ has the same distribution as
$\frac{\sigma}{\sqrt{H+1}} W$. Next, from the fact that for all $n$,
$M^{n,j}$ is a martingale with respect to the filtration generated by
$A^{n,j}$ and $M^{n,j}$, and its predictable variation is uniformly bounded (follows from (\ref{variation})), we obtain that the limit process
$W^{(j)}$ is a martingale with respect to the natural filtration
generated by $W^{(j)}$ and $A^{(j)}$ (for details see Chapter 9 in \cite{JS}).
Thus, from (\ref{4.9}) it follows that
\begin{equation}\label{4.10}
\langle W^{(j)}+A^{(j)}\rangle _t=\langle W^{(j)}\rangle _t=\frac{\sigma^2 t}{H+1} \ \  \forall t\in [0,1], \ \ \mbox{a.s.}
\end{equation}
Now, we arrive to the final step of the proof.
Without loss of generality (by passing to a subsequence), we assume that the sequence
of the joint distributions
$$\left(\left(\{M^{n,0}_{[nt]}\}_{t=0}^1,...,\{M^{n,H}_{[nt]}\}_{t=0}^1,\{A^{n,0}_{[nt]}\}_{t=0}^1,....,\{A^{n,H}_{[nt]}\}_{t=0}^1\right),\mathbb Q_n\right)$$
converges.
Observe that
$$\frac{\sigma}{\sqrt n}\sum_{i=1}^{[nt]}\xi_i=O(n^{-1/2})+\sum_{j=0,1,...,H}\left(A^{n,j}_{[nt]}+M^{n,j}_{[nt]}\right).$$
This together with (\ref{4.9+}) and the weak convergence
$$\left(\left\{\frac{\sigma}{\sqrt n}\sum_{i=1}^{[nt]}\xi_i\right\}_{t=0}^1,\mathbb Q_n\right)\Rightarrow \{N_t-\langle N\rangle_t/2\}_{t=0}^1$$
gives that the distribution of $\sum_{j=0,1,...,H}\left(A^{(j)}+W^{(j)}\right)$
equals to the distribution of
 $N-\langle N\rangle/2$. Moreover, from the equality
$\langle N\rangle\equiv \langle N-\langle N\rangle/2\rangle$  we deduce that, in (\ref{4.6}) we can replace $N$ with
$\sum_{j=0,1,...,H}\left(A^{(j)}+W^{(j)}\right)$.

Finally, from (\ref{4.10}) and the simple
inequality
$$(\sum_{i=1}^{H+1} a_i)^2\leq (H+1)\sum_{i=1}^{H+1} a^2_i, \ \ a_1,...,a_{H+1}\in\mathbb R$$ we get that for any $T_1<T_2$
\begin{eqnarray*}
&\langle \sum_{j=0,1,...,H}\left(A^{(j)}+W^{(j)}\right)\rangle _{T_2}-\langle \sum_{j=0,1,...,H}\left(A^{(j)}+W^{(j)}\right)\rangle _{T_1}\leq\\
&(H+1)\sum_{j=0,1,...,H}\left(\langle A^{(j)}+W^{(j)}\rangle _{T_2}-\langle A^{(j)}+W^{(j)}\rangle _{T_1}\right)\leq \sigma^2 (H+1)(T_2-T_1)
\end{eqnarray*}
and (\ref{4.6}) follows.
\end{proof}

\subsection{Proof of the Lower Bound}
Recall from Section \ref{sec:7} the Brownian probability space
$(\Omega^W, \mathcal{F}^{W}, \mathbb P^{W})$
and the sets $\Gamma,\Gamma_c$. Introduce the sets
$$\Gamma^H:=\{\alpha\in\Gamma: \alpha\leq\sigma\sqrt{H+1} \ \ \mathbb P^W\otimes dt \  \mbox{a.s.}\}$$ and
$$\Gamma^H_c:=\{\alpha\in\Gamma_c: \alpha\leq\sigma\sqrt{H+1} \ \ \mathbb P^W\otimes dt \ \mbox{a.s.}\}.$$

By applying a randomization
technique similar to Lemma 7.2 in \cite{DS} we obtain
\begin{equation}\label{4.100}
\sup_{\mathbb Q\in\mathcal Q^H}\mathbb E_{\mathbb Q}[F(\mathbb S)]=\sup_{\alpha\in\Gamma^H}\mathbb E_{\mathbb P^W}[F(S^{(\alpha)})]
\end{equation}
where, recall $S^{(\alpha)}$ is given before Proposition \ref{prop7.1}.

Next, standard density arguments (see Lemma 3.4 in \cite{BDP} for the case $d=1$) imply that
\begin{equation}\label{4.101}
\sup_{\alpha\in\Gamma^H}\mathbb E_{\mathbb P^W}[F(S^{(\alpha)})]=\sup_{\alpha\in\Gamma^H_c}\mathbb E_{\mathbb P^W}[F(S^{(\alpha)})].
\end{equation}
From (\ref{dual}), Assumption \ref{asm4.1} and
(\ref{4.100})--(\ref{4.101}), it follows that in order to prove the lower bound, i.e. the inequality
$$\lim \inf_{n\rightarrow\infty} V_n\geq \sup_{\mathbb Q\in\mathcal Q^H}\mathbb E_{\mathbb Q}[F(\mathbb S)]$$
it remains to establish the following lemma.
\begin{lem}
For any $\alpha\in \Gamma^H_c$ there exists a sequence of probability measures
${\mathbb Q}_n\in\mathcal Q^H_n$ such that we have the weak convergence
$\left(\{S^{(n)}_{[nt]}\}_{t=0}^1,{\mathbb Q}_n\right)\Rightarrow \{ S^{(\alpha)}_t\}_{t=0}^1$.
\end{lem}
\begin{proof}
Choose $\alpha\in\Gamma^H_c$ and let $0=t_0<t_1<\dots<t_J=1$, $\epsilon>0$,
$\rho_j:\mathbb R^{j} \rightarrow [\epsilon,\sigma\sqrt{H+1}]$, $j=0,1,...,J-1$
such that (\ref{4.102}) holds true.
The proof will be done in three steps.\\
\textbf{Step I:}
In this step we construct the probability measures ${\mathbb Q}_n$, $n\in\mathbb N$.
Fix $n\in\mathbb N$.
For any $j=0,1,...,J-1$ consider the interval
$I_j:=[[n t_j],[n t_{j+1}])$ and
let $U^{(j)}_1,U^{(j)}_2,U^{(j)}_3\subset I_j$ be disjoint sets such that
$U^{(j)}_1\bigcup U^{(j)}_2\bigcup U^{(j)}_3=I_j$. We allow $U^{(j)}_1,U^{(j)}_2,U^{(j)}_3$
to be random sets in the sense that they can depend on $S^{(n)}_1,....,S^{(n)}_{[nt_j]}$ (i.e. can depend on the
stock prices up to the moment $[nt_j]$).

From the simple estimate
$S^{(n)}_{k+1}= S^{(n)}_k\left(1+\frac{\sigma \xi_{k+1}}{\sqrt n}+O(n^{-1})\right)$
it follows that
for sufficiently large $n$ we can find a probability measure ${\mathbb Q}_n$
such that
\begin{eqnarray*}
&\mathbb E_{{\mathbb Q}_n}\left( S^{(n)}_{k+1}- S^{(n)}_k|\mathcal F_k\right)=0 \ \ \forall k\in U^{(j)}_1\\
&\mathbb E_{{\mathbb Q}_n}\left( S^{(n)}_{k+1}- S^{(n)}_k|\mathcal F_k\right)=
\left(1-\frac{1}{\sqrt n}\right)( S^{(n)}_{k}- S^{(n)}_{k-1})\ \ \forall k\in U^{(j)}_2\\
&\mathbb E_{{\mathbb Q}_n}\left(S^{(n)}_{k+1}-S^{(n)}_k|\mathcal F_k\right)=
-\left(1-\frac{1}{\sqrt n}\right)(S^{(n)}_{k}- S^{(n)}_{k-1})\ \ \forall k\in U^{(j)}_3.
\end{eqnarray*}
Now, we explain how to choose the sets $U^{(j)}_1,U^{(j)}_2,U^{(j)}_3$, $j=0,1,...,J-1$.
For any $j$
divide the interval $I_j$ into $[\sqrt n (t_{j+1}-t_j)]$ blocks of
the same number $[\sqrt n]$ of successive time points.
Because we restrict ourself to integer blocks it might happen that these blocks cover the interval
$I_j$, besides of $O(n^{-1/2})$ successive time points which lie in the right end of the interval.
We define all
the not covered time points to be elements of the set $U^{(j)}_1$.

Next, on each of the $[\sqrt n (t_{j+1}-t_j)]$ blocks of $[\sqrt n]$ successive time points we apply the following procedure.
First, each block
is divided into $[\sqrt n/(2H+2)]$ blocks of
$(2H+2)$ points, and again the missing points are defined to be elements of the set $U^{(j)}_1$.
Introduce the random variable
\begin{equation}\label{4.final+}
A^{(n)}_j=\frac{\rho^2_j( S^{(n)}_{[n t_1]},..., S^{(n)}_{[n t_j]})\sqrt n}{2\sigma^2 (H+1)^2}.
\end{equation}
In the first $A^{(n)}_j$ blocks,
for each block $\{k,k+1,...,k+2H+1\}$, the points
$k,k+H+1$ are sent to the set $U^{(j)}_1$ and the rest of the points are sent to the set
$U^{(j)}_2$.
In the remaining $[\sqrt n/(2H+2)]-A^{(n)}_j$ blocks (notice that $A^{(n)}_j\leq [\sqrt n/(2H+2)]$), for
each block $\{k,k+1,...,k+2H+1\}$, the points
$k,k+2,k+4,...$ are sent to the set $U^{(j)}_1$ and the rest of the points are sent to $U^{(j)}_3$.\\
\textbf{Step II:}
In this step we derive essential properties of the construction.
For any $n\in\mathbb N$ define the stochastic processes
$ {M}^{(n)}=\{{M}^{(n)}_k\}_{k=0}^n$
and $ {N}^{(n)}=\{{ N}^{(n)}_k\}_{k=0}^n$
by
\begin{eqnarray*}
&M^{(n)}_k= S^{(n)}_{\max\{m\leq k: m\in U^{(j)}_1\}}, \ \ [nt_j]\leq k< [nt_j]+1, \ \ j=0,1,...,J-1,\\
& N^{(n)}_k=\sum_{i=0}^{k-1}\frac{M^{(n)}_{i+1}-  M^{(n)}_{i}}{\Theta^{(n)}_{i} M^{(n)}_{i}},\ \ k=0,1,...,n
\end{eqnarray*}
where $\Theta^{(n)}_i=\rho_j( S^{(n)}_{[nt_1]},..., S^{(n)}_{[nt_j]})$ for $[nt_j]\leq i< [nt_{j+1}]$
and we
set $ M^{(n)}_n= M^{(n)}_{n-1}$. Introduce the filtration $\mathcal G^{(n)}=\{\mathcal G^{(n)}_k\}_{k=0}^n$
by
$$\mathcal G^{(n)}_k=\sigma\{ M^{(n)}_1,..., M^{(n)}_k,
 S^{(n)}_{[nt_1]},..., S^{(n)}_{[nt_j]}\} \ \  [nt_j]\leq k< [nt_{j+1}], \ \ j=0,1,...,J-1.$$
From the definition of the sets $U^{(j)}_1,U^{(j)}_2,U^{(j)}_3$ it follows that
$M^{(n)}$ is a ${\mathbb Q}_n$ martingale with respect to
the filtration $\mathcal G^{(n)}$.
Hence ${N}^{(n)}$ is also a martingale.
Next, we observe that for any sequence of $H+1$ successive time points there is at least one point which belongs to $U^{(j)}_1$ for some $j$.
This together with the definition of the
sets $U^{(j)}_2,U^{(j)}_3$ implies that (\ref{3.1}) holds true.
Thus ${\mathbb Q}_n\in\mathcal Q^H_n$ and
\begin{equation}\label{4.103+}
| M^{(n)}_k- S^{(n)}_k|\leq \frac{\tilde C}{\sqrt n}  S^{(n)}_k \ \ \forall k\leq n
\end{equation}
for some constant $\tilde C>0$.

Lemma 3.3 in \cite{BDP} and (\ref{4.103+}) implies that the sequence
$\left(\{ S^{(n)}_{[nt]}\}_{t=0}^1,{\mathbb Q}_n\right)$, $n\in\mathbb N$ is tight
and any cluster point is a strictly positive, continuous martingale.
Moreover, the martingales $M^{(n)}$, $n\in\mathbb N$ satisfy (\ref{4.5+}).
We need to prove that any cluster point satisfies (\ref{4.103}) (recall the uniqueness property of (\ref{4.103})).
Thus, choose a subsequence (we still index it by $n$)
$\left(\{S^{(n)}_{[nt]}\}_{t=0}^1,\tilde{\mathbb Q}_n\right)$
which converges to a martingale $M$.
We will apply the stability results from \cite{DP}. First, (\ref{4.5+}) implies that
the sequence $M^{(n)}$, $n\in\mathbb N$ satisfies Condition A in \cite{DP}.
Thus, from Theorem 4.1 in \cite{DP} we conclude (recall that $\rho_j\geq \epsilon>0$ for all $j$).
\begin{equation}\label{4.104}
\left(\{ N^{(n)}_{[nt]}\}_{t=0}^1,{\mathbb Q}_n
\right)\Rightarrow \left\{\sum_{j=0}^{J-1} \int_{t_j\wedge t }^{t_{j+1\wedge t }}\frac{dM_u}{\rho_j(M_{t_1},...,M_{t_j})M_u}\right\}_{t=0}^1.
\end{equation}
\\
\textbf{Step III:}
In view of (\ref{4.104}) in order to complete the proof we need
to show that
the sequence $\left(\{ N^{(n)}_{[nt]}\}_{t=0}^1,\mathbb Q_n\right)$, $n\in\mathbb N$
converges to the standard Brownian motion. We will apply the Martingale Central Limit theorem: Proposition 1 in \cite{R}.
Clearly, $N^{(n)}_k- N^{(n)}_{k-1}=O(n^{-1/2})$ for all $k\leq n$. Thus, it remains to
show that the predictable variation of
$ N^{(n)}$, $n\in\mathbb N$ satisfies
\begin{equation}\label{4.fin}
 \left(\{ \langle N^{(n)}\rangle_{[nt]}\}_{t=0}^1,{\mathbb Q}_n\right)\Rightarrow \{t\}_{t=0}^1.
\end{equation}

Fix $n\in\mathbb N$ and
consider an interval $[nt_j,n t_{j+1}]$ for some $j$.
Recall the construction in Step I and choose a block of
$[\sqrt n]$ successive time points.
We notice that in the first $A^{(n)}_j$ blocks, for any block
$\{k,k+1,...,k+2H+1\}$ we have
\begin{eqnarray*}
&S^{(n)}_k= M^{(n)}_k=M^{(n)}_{k+1}=...= M^{(n)}_{k+H}\\
&\mbox{and} \ \  S^{(n)}_{k+H+1}=M^{(n)}_{k+H+1}=...= M^{(n)}_{k+2H+1}.
\end{eqnarray*}
Moreover, from the definition of the set $U^{(j)}_2$ we get
\begin{eqnarray*}
&{\mathbb Q}_n\left( S^{(n)}_{k+H+1}= S^{(n)}_{k}\exp\left(\pm\sigma (H+1)n^{-1/2}\right)
|\mathcal F_k\right)=1-O(n^{-1/2}) \ \ \mbox{and}\\
&{\mathbb Q}_n\left( S^{(n)}_{k+2H+2}= S^{(n)}_{k+H+1}\exp\left(\pm\sigma (H+1)n^{-1/2}\right)
|\mathcal F_{k+H+1}\right)=1-O(n^{-1/2}).
\end{eqnarray*}
We conclude that for any block
$\{k,k+1,...,k+2H+1\}$ (of the first $A^{(n)}_j$ blocks)
\begin{eqnarray}\label{4.final1}
&\mathbb E_{{\mathbb Q}_n}\left(( N^{(n)}_{i+1}- N^{(n)}_i)^2|\mathcal G^{(n)}_i\right)=
\frac{\sigma^2(H+1)^2}{n\rho^2_j( S^{(n)}_{[nt_1]},..., S^{(n)}_{[nt_j]})}+O(n^{-3/2})\\
& \ i\in\{k+H,k+2H+1\} \ \ \mbox{and} \ \
\mathbb E_{{\mathbb Q}_n}\left((N^{(n)}_{i+1}-N^{(n)}_i)^2|\mathcal G^{(n)}_i\right)=0 \  \mbox{otherwise}.\nonumber
\end{eqnarray}
On the other hand for any block
$\{k,k+1,...,k+2H+1\}$ of the
remaining  $[\sqrt n  /(2H+2)]-A^{(n)}_j$ blocks we have
$ M^{(n)}_i= S^{(n)}_i$ for $i=k,k+2,...$
and $ M^{(n)}_i= S^{(n)}_{i-1}$ for $i=k+1,k+3,...$.
Furthermore, from the definition of the set $U^{(j)}_3$ it follows that
for any $i\in\{k,k+2,...\}$
$${\mathbb Q}_n\left({ S^{(n)}_{i+2}}={ S^{(n)}_{i}}\right)=1-O(n^{-1/2}).$$
We conclude that for any block
$\{k,k+1,...,k+2H+1\}$ (of the first remaining  $[\sqrt n  /(2H+2)]-A^{(n)}_j$ blocks)
\begin{equation}\label{4.final2}
\mathbb E_{{\mathbb Q}_n}\left(( N^{(n)}_{i+1}- N^{(n)}_i)^2|\mathcal G^{(n)}_i\right)=O(n^{-3/2}) \ \ \forall i.
\end{equation}
Finally, choose $j$ and $t_j<T_1<T_2<t_{j+1}$. From (\ref{4.final+}) and (\ref{4.final1})--(\ref{4.final2}) we obtain
\begin{eqnarray*}
&\langle  N^{(n)}\rangle _{[nT_2]}-\langle  N^{(n)}\rangle _{[nT_2]}=\frac{2 A^{(n)}_j n(T_2-T_2)}{[\sqrt n]}
\frac{\sigma^2(H+1)^2}{n\rho^2_j( S^{(n)}_{[nt_1]},..., S^{(n)}_{[nt_j]})}+O(n^{-1/2})=\\
&(T_2-T_1)+O(n^{-1/2})
\end{eqnarray*}
and (\ref{4.fin}) follows.
\end{proof}

\section*{Acknowledgments}
This research was supported by the ISF grant 160/17.

\bibliographystyle{spbasic}

\end{document}